\documentclass[11pt]{article}
\usepackage[a4paper,margin=2cm]{geometry}
\usepackage{microtype,bm,cite,amssymb,url,hyperref,amsthm}

\usepackage{inkfrag}

\definecolor{blueblack}{rgb}{0,0,.7}
\newcommand{\emphdef}[1]{%
  \textcolor{blueblack}{%
    \textbf{\emph{#1}}%
  }%
}
\newcommand{\CC}{\ensuremath{\bar{C}}}
\newcommand{\KK}{\ensuremath{\bar{K}}}
\newcommand{\ZZ}{\ensuremath{\mathbb{Z}}}
\newcommand{\set}[1]{\ensuremath{\{#1\}}}

\theoremstyle{plain}  
\newtheorem{theorem}{Theorem}[section]
\newtheorem{lemma}[theorem]{Lemma}
\newtheorem{proposition}[theorem]{Proposition}

\theoremstyle{definition}
\newtheorem{definition}[theorem]{Definition}

\title{Multicuts in Planar and Bounded-Genus Graphs\\ with Bounded Number of
  Terminals\thanks{Work supported by the French ANR Blanc project
    ANR-12-BS02-005 (RDAM).}%
}

\author{\'Eric Colin de Verdi\`ere\thanks{CNRS, Laboratoire d'Informatique
    Gaspard Monge, Marne-la-Vall\'ee, France.
    \texttt{eric.colindeverdiere}\texttt{@u-pem.fr}.  Part of this work was
    done while the author was at CNRS, D\'epartement d'informatique,
    \'Ecole normale sup\'erieure, Paris, France.}}

\date{\today}

\begin{document}

\maketitle

\begin{abstract}
  Given an undirected, edge-weighted graph $G$ together with pairs of
  vertices, called pairs of \emph{terminals}, the \emph{minimum multicut
    problem} asks for a minimum-weight set of edges such that, after
  deleting these edges, the two terminals of each pair belong to different
  connected components of the graph.  Relying on topological techniques, we
  provide a polynomial-time algorithm for this problem in the case where
  $G$ is embedded on a fixed surface of genus~$g$ (e.g., when $G$ is
  planar) and has a fixed number~$t$ of terminals.  The running time is a
  polynomial of degree~$O\big(\sqrt{g^2+gt}\big)$ in the input size.

  In the planar case, our result corrects an error in an extended abstract
  by Bentz [Int.\ Workshop on Parameterized and Exact Computation,
  109--119, 2012].  The minimum multicut problem is also a generalization
  of the \emph{multiway cut problem}, a.k.a.\ \emph{multiterminal cut
    problem}; even for this special case, no dedicated algorithm was known
  for graphs embedded on surfaces.
\end{abstract}

\section{Introduction}\label{S:intro}

The minimum cut problem is one of the most fundamental problems in
combinatorial optimization (see Schrijver~\cite{s-hcot-05} for a
fascinating historical account), originally introduced in relation to
railway transshipment problems during the cold war~\cite{hr-fmern-55}.  In
this context, the railway network is modeled by a planar graph, each edge
having a weight (its capacity), and the goal is to compute the
minimum-weight set of edges that need to be removed to disconnect two given
vertices of the network, the source and destination for a given commodity.
While countless generalizations of this problem have been studied, we are
interested here in two natural extensions:
\begin{enumerate}
\item What if there are several commodities, corresponding to different
  source and destination pairs?  In other words, we are studying an instance
  of the minimum multicut problem: Given several pairs of source and
  destination vertices, how to find a minimum-weight set of edges to
  disconnect every destination vertex from its corresponding source?
\item What if the network is not planar, but includes a few tunnels and
  bridges?  In other words, what happens if the graph is embedded, not in
  the plane, but on some fixed surface?
\end{enumerate}

More formally, let $G=(V,E)$ be an undirected graph.  Furthermore, let $T$
be a subset of vertices of~$G$, called \emphdef{terminals}, and let $R$ be
a set of unordered pairs of vertices in~$T$, called \emphdef{terminal
  pairs}.  A subset $E'$ of~$E$ is a \emphdef{multicut} (with respect
to~$(T,R)$) if for every terminal pair $\set{t_1,t_2}\in R$, the vertices $t_1$
and~$t_2$ are in different connected components of the graph $(V,E\setminus
E')$.  In the \emphdef{minimum multicut problem} (also known as the
\emph{minimum multiterminal cut problem}), we assume in addition that $G$
is positively edge-weighted, and the goal is to find a multicut of minimum
weight.  We prove that this problem is polynomial-time solvable if $G$ is
embedded on a fixed surface and the number~$t$ of terminals is fixed.  More
precisely:
\begin{theorem}\label{T:main}
  Assume that $G$ is cellularly embedded on a surface (orientable or not)
  of Euler genus~$g$.  Then the minimum multicut problem can be solved in
  ${(g+t)^{O(g+t)}}\*{n^{O\big(\sqrt{g^2+gt}\big)}}$ time, where $t=|T|$ is
  the number of terminals and $n$~is the number of edges of~$G$.
\end{theorem}
This is the first polynomial-time algorithm for this purpose, even when
specialized to either the multiway cut problem (see below for details) or
the planar version.  Moreover, the $n^{O(\sqrt{t})}$ dependence in the
number of terminals is unavoidable, assuming the Exponential Time
Hypothesis~\cite{m-tlbpm-12}, even in these two special cases.

\subsection*{Comparison with Former Work}

Many instances of the \textbf{minimum multicut} problem are hard, even in
very restricted cases.  In particular, it is NP-hard in unweighted binary
trees~\cite{cfr-mugdb-03} and unweighted
stars~\cite[Theorem~3.1]{gvy-pdaai-97}, and even APX-hard in the latter
case.  In the case where the number of pairs of terminals is fixed and at
least three, Dahlhaus et al.~\cite{djpsy-cmc-94} have proved that the
problem is APX-hard in general graphs; nonetheless, it becomes
polynomial-time solvable for bounded-treewidth graphs, as proved by
Bentz~\cite[Theorem~1]{b-cmpbt-08}, and fixed parameter tractable in the
size of the solution for unweighted graphs~\cite{bdt-mf-11}.  The problem
is even harder for directed graphs~\cite{b-cmpbt-08}.

\smallskip

In the case where the graph is \textbf{planar}, the number of terminals is
fixed, and they all lie on the outer face, Bentz~\cite{b-sampg-09} has
given a polynomial-time algorithm for the minimum multicut problem.  More
recently~\cite{b-ptapm-12}, he has announced an algorithm for the same
case, but removing the condition that the terminals lie on the outer face.
Unfortunately, his proof has several flaws, leaving little hope for repair
(see Appendix~\ref{A:bentz}).  We give a faster algorithm that also works
for graphs on arbitrary surfaces.

\smallskip

A special case that is somewhat more tractable is the \textbf{multiway cut
  problem} (a.k.a.\ the \emph{multiterminal cut problem}); this is the case
where the set of pairs of terminals comprises all possible pairs of
distinct vertices in the set of terminals $T\subset V$.  In the planar
case, Dahlhaus et al.~\cite{djpsy-cmc-94} have proved that it is still
NP-hard, but Bateni et al.~\cite{bhkm-ptasp-12} have given a
polynomial-time approximation scheme.  Again in the planar case, the
problem is also polynomial-time solvable if the number of terminals is
fixed, as proved in the early 1990s~\cite{djpsy-cmc-94,h-pmcp-98}.  In
stark contrast, the complexity of the multicut problem has remained open
until now, although it is a very natural generalization of the multiway cut
problem (the multicut problem is ``dual'' to the multicommodity flow
problem, largely studied~\cite[Chapters 70--76]{s-cope-03}).

More recently, Klein and Marx have shown that the planar multiway cut
problem can be solved in $2^{O(t)}\*n^{O(\sqrt{t})}$ time (where $n$ is the
complexity of the graph)~\cite{km-spktc-12}; Marx has proved that the
$n^{O(\sqrt{t})}$-dependence is the best one could hope for, assuming the
Exponential Time Hypothesis (ETH)~\cite{m-tlbpm-12}.  Our algorithm is more
general since it deals with multicut, not multiway cut, and works on
arbitrary surfaces; its running time, for fixed genus, is
$t^{O(t)}\*n^{O(\sqrt{t})}$; while the $t^{O(t)}$ factor is slightly worse
than the $2^{O(t)}$ of Klein and Marx, the second factor is the same, and
optimal unless ETH is false.  Since approximability in the planar case is
very different for multicut and multiway cut, our result is surprising,
since it shows that, as far as exact computation is concerned, both are
(essentially) equally hard.

As a side note, it is easy to see that, to solve the multicut problem, it
suffices to guess the partition of the terminals into connected components,
and to solve a multiway cut problem in a higher genus surface.  However,
this reduction is not useful here, since we would get a worse dependence
in~$t$ if we were to prove the result for multiway cut on surfaces and use
this reduction.

\smallskip

On the other hand, graph algorithms dedicated to \textbf{graphs embedded on
  a fixed surface} have flourished during the last decade.  One reason is
that many graphs arising in geometric settings are naturally embedded on a
surface; another one is that the family of graphs embeddable on a fixed
surface is closed under minor, and such studies can be seen as a first step
towards efficient algorithms for minor-closed families.  This line of
research is also justified by the fact that testing whether a graph of
complexity~$n$ embeds on a surface of genus~$g$ can be done efficiently in
the complexity of the graph, namely, in $2^{O(g)}n$ time~\cite{m-ltaeg-99}.
However, the history of flow and cut problems for graphs embedded
on surfaces is rather short.  Chambers et al.~\cite{cen-mcshc-09} have
given an algorithm to compute a minimum cut in a graph embedded on an
orientable surface of genus~$g$ that runs in $g^{O(g)}n\log n$ time; a very
different algorithm by Erickson and Nayyeri~\cite{en-mcsnc-11} runs in
$2^{O(g)}n\log n$ time.  Thus, the minimum cut problem is solvable in
near-linear time on graphs embeddable on a fixed orientable surface.
Algorithms are also available for computing global minimum
cuts~\cite{efn-gmcse-12} and maximum flows~\cite{cen-hfcc-12} for graphs on
surfaces.  To our knowledge, we present here the first algorithm for the
minimum multicut problem (or even the multiway cut problem) for
surface-embedded graphs.  The main tool for flows and cuts on
surfaces~\cite{cen-hfcc-12,cen-mcshc-09,en-mcsnc-11,efn-gmcse-12} is
homology, which is the appropriate algebraic formalism for dealing with
graphs separating two given vertices, but it appears to be insufficient in
the multicommodity case.

\subsection*{Overview and Discussion of Proof Techniques}

The strategy for proving Theorem~\ref{T:main} is the following.  In
Section~\ref{S:cross-metric}, we first show that a multicut corresponds, in
a dual sense, to a graph drawn on~$S$ that separates all pairs of
terminals; such a graph will be called a \emph{multicut dual}.  Moreover,
if the multicut is minimum, then this multicut dual is as short as
possible, when distances on~$S$ are measured using the \emph{cross-metric}:
namely, the sum of the weights of the edges of~$G$ crossed by the multicut
dual is minimum.  The topological structure of the multicut dual can be
described suitably after we cut the surface open into a disk with all
terminals on its boundary (Section~\ref{S:cutgraph}).  We then show that
this structure is constrained (Section~\ref{S:struct}), that we can
enumerate its various possibilities (Section~\ref{S:enum}), and (roughly)
that, for each of these topologies, we can compute a shortest multicut dual
with that topology efficiently (Section~\ref{S:final}).

At a high level, our approach follows a similar pattern to Klein and
Marx~\cite{km-spktc-12}, since they also rely on enumerating the various
candidate topologies for the dual solution and find the optimum solution
for each topology separately.  This strategy is also present in Dahlhaus et
al.~\cite{djpsy-cmc-94} and, in a different context, in Chambers et
al.~\cite{ccelw-scsh-08}, which was our initial source of inspiration.  The
details are, however, rather different.

Indeed, Klein and Marx~\cite{km-spktc-12} need a reduction to the
biconnected case~\cite[Section~3]{km-spktc-12}, which is shortcut in our
approach.  Also, the structural properties that we develop for the minimum
multicut problem are more involved than the known ones for the multiway cut
problem; indeed, the solution, viewed in the dual graph, has less structure
in the multicut problem than in the multiway cut problem: in the multiway
cut case, it has as many faces as terminals, and thus many cycles, whereas
for the minimum multicut problem, the possible topologies are more diverse
(e.g., an optimal solution could be a single dual cycle).

Chambers et al.~\cite{ccelw-scsh-08} have developed related techniques for
computing a shortest splitting cycle, which have been subsequently reused
for other topological and computational problems in planar or surface
cases~\cite{en-sncwp-11,cen-mcshc-09,f-sntcd-13}.  A key difference,
however, is that we extend the method to work with graphs instead of paths
or cycles, which makes the arguments more complicated.  We need to encode
precisely the locations of the vertices and edges of the multicut dual with
respect to the cut graph; the cross-metric setting is very convenient for
this purpose, since it avoids successive transformations of the input dual
graph to mimic cutting along another graph, as done by Klein and
Marx~\cite{km-spktc-12}.

Moreover, our approach also relies on other techniques from computational
topology, in particular, homology techniques developed for the single
commodity minimum cut problem~\cite{cen-mcshc-09}, homotopy techniques for
shortest homotopic paths~\cite{ce-tnpcs-10,k-csntc-06}, and treewidth
techniques for the surface case~\cite{e-dgteg-03}.

Finally, we remark that we are not aware of any significantly simpler proof
for the planar case: The construction of Section~\ref{S:cutgraph} can be
simplified, and Lemma~\ref{L:treewidth} is standard in that case, but the
overall strategy would be the same.

\section{Preliminaries}\label{S:prelim}

We recall here standard definitions on the topology of surfaces.  For
general background on topology, see for example Stillwell~\cite{s-ctcgt-93}
or Hatcher~\cite{h-at-02}.  For more specific background on the topology of
surfaces in our context, see recent articles and surveys on the same
topic~\cite{ce-tnpcs-10,c-tags-12}.

In this article, $S$ is a compact, connected \emphdef{surface} without
boundary; $g$ denotes its Euler genus.  Thus, if $S$ is orientable,
$g\geq0$ is even, and $S$ is (homeomorphic to) a sphere with $g/2$ handles
attached; if $S$ is non-orientable, $S$ is a sphere with $g\geq1$ disjoint
disks replaced by M\"obius strips.

We consider paths drawn on~$S$.  A \emphdef{path}~$p$ is a continuous map
from $[0,1]$ to~$S$; a \emphdef{loop} is a path~$p$ such that its two
endpoints coincide: $p(0)=p(1)$.  A path is \emphdef{simple} if it is
one-to-one (except, of course, that its endpoints $p(0)$ and~$p(1)$ may
coincide).  We thus emphasize that, contrary to the standard terminology in
graph theory, paths may self-intersect.  A simple loop is
\emphdef{two-sided} if it has a neighborhood homeomorphic to an annulus;
otherwise, it has a neighborhood homeomorphic to a M\"obius strip (which
implies that the surface is non-orientable), and is \emphdef{one-sided}.

All the graphs considered in this article may have loops and multiple
edges.  A \emphdef{drawing} of a graph~$G$ on~$S$ maps the vertices of~$G$
to points on~$S$ and the edges of~$G$ to paths on~$S$ whose endpoints are
the images of the incident vertices.  An \emphdef{embedding} of~$G$ is a
``crossing-free'' drawing: The images of the vertices are pairwise
distinct, and the image of each edge is a simple path intersecting the
image of no other vertex or edge, except possibly at its endpoints.  We
will mostly consider graph embeddings on~$S$.  A \emphdef{face} of an
embedding of~$G$ on~$S$ is a connected component of $S$ minus (the image
of)~$G$.  A graph is \emphdef{cellularly embedded} on~$S$ if every face of
the graph is an open disk.  A \emphdef{cut graph} of~$S$ is a graph~$G$
embedded on~$S$ whose unique face is a disk.  (In particular, a cut graph
has no isolated vertex, unless it is reduced to a single vertex and $S$ is
a sphere---this case does not occur in this paper.)  \emphdef{Euler's
  formula} states that, if $G$ is a graph cellularly embedded on~$S$ with
$v$ vertices, $e$ edges, and $f$ faces, then $v-e+f=2-g$.

Algorithmically, we can store graphs cellularly embedded on~$S$ by their
\emph{combinatorial map}, which essentially records the cyclic ordering of
the edges around each vertex; there are efficient data structures for this
purpose~\cite{e-dgteg-03,l-gem-82}.

\section{The Cross-Metric Setting}\label{S:cross-metric}

In this section, we prove that a minimum multicut corresponds, in an
appropriate sense, to a shortest graph satisfying certain properties.

We say that a graph~$H$ embedded on~$S$ is in \emphdef{general position}
with respect to our input graph~$G$ if there are finitely many intersection
points between $G$ and~$H$, and each such point corresponds to a crossing
between an edge of~$G$ and an edge of~$H$.  The \emphdef{length} of~$H$ is
the sum, over all crossing points between $G$ and~$H$, of the weight of the
corresponding edge of~$G$.  Note that an edge of~$H$ can cross an edge
of~$G$ several times, and in such cases, the length of the edge of~$H$ is
computed by taking into account the multiplicity of intersections.
In other words, $G$ is now seen as a graph that
provides a discrete (or \emph{cross-metric}) distance function
on~$S$~\cite{ce-tnpcs-10}.  Algorithmically, we can store a graph~$H$ in
general position with respect to~$G$ by recording the combinatorial map of
the \emph{overlay} of $G$ and~$H$, obtained by adding vertices at each
intersection point between $G$ and~$H$ and subdividing edges of~$G$
and~$H$.

\emph{In the following, unless noted otherwise, \textbf{all graphs drawn
    on~$\bm{S}$ will be in general position with respect to~$\bm{G}$}.
  Moreover, whenever we consider distances between two points in~$S$ (not
  lying on~$G$) or lengths of paths in~$S$, we implicitly consider them in
  the above cross-metric sense.}  In some clearly mentioned cases below
(see Proposition~\ref{P:cutgraph}), we will need to consider paths~$p$ that
are in general position with respect to~$G$, except that some of their
endpoints may lie on~$G$.  In such cases, the endpoints of~$p$ are not
taken into account for determining the length of~$p$.

A \emphdef{multicut dual} is a graph~$C$ embedded on~$S$ such that, for
every pair $\set{t_1,t_2}\in R$, the vertices $t_1$ and~$t_2$ are in different
faces of~$C$.  As the terminology suggests, we have the following
proposition, which will guide our approach.
\begin{proposition}\label{P:dual}
  Let $C$ be a shortest multicut dual.  Then the set $E'$ of edges of~$G$
  crossed at least once by~$C$ is a minimum multicut.
\end{proposition}
\begin{proof}
  $E'$ is a multicut, because any path in~$G$ connecting a pair of
  terminals $\set{t_1,t_2}\in R$ must cross the multicut dual, and thus use
  one edge in~$E'$.  Moreover, the weight of~$E'$ is at most the length
  of~$C$.

  To prove that $E'$ is a multicut of minimum weight, it suffices to prove
  that, for every multicut~$E''$ with weight~$w$, there exists a multicut
  dual of length~$w$.  Consider the \emph{dual graph}~$G^*$ that has one
  vertex inside each face of~$G$ and one edge~$e^*$ crossing each edge~$e$
  of~$G$.  The subgraph of~$G^*$ with edge set $\set{e^*\mid e\in E''}$
  forms a multicut dual whose length is~$w$.
\end{proof}

As a side remark, it follows that the minimum multicut problem can be seen
as a discrete version of the following topological problem: Given a metric
surface~$S$ with boundary, and a set~$R$ of pairs of boundary components,
compute a shortest graph on~$S$ that separates every pair of boundaries
in~$R$.  We are exactly solving this problem in the realm of cross-metric
surfaces.

\section{Planarization}\label{S:cutgraph}

Our algorithm starts by computing a cut graph~$K$ of~$S$ passing through
all the terminals.  We will also need some structural properties for~$K$,
detailed in the following proposition.  If $S$ is the sphere (equivalently,
if $G$ is planar), we could take for~$K$ a shortest spanning tree of~$T$
(with respect to the cross-metric setting), which can be obtained using a
simple modification of any algorithm for computing minimum spanning trees
in graphs.  For the general case, we use a known construction, a so-called
\emph{greedy system of arcs}~\cite{ccelw-scsh-08}.  The following
proposition summarizes the properties that we will use.

\begin{proposition}\label{P:cutgraph}
  In $O(n\log n+(g+t)n)$ time, we can compute a cut graph~$K$ on~$S$, whose
  $O(g+t)$ vertices contain~$T$, and with $O(g+t)$ edges, each of which is
  a shortest path on~$S$.  Some vertices of~$K$ may lie on~$G$ (either on
  vertices or on edges).
\end{proposition}
\begin{proof}
  We temporarily remove a small disk containing each terminal, the boundary
  of which crosses each edge of~$G$ incident to that terminal exactly once,
  and crossing no other edge of~$G$.  This yields a surface with
  boundary~$S'$ that is naturally a cross-metric surface, because the
  intersection of the image of the graph~$G$ with~$S'$ is also a graph~$G'$
  embedded on~$S'$.  On that surface, we compute a \emph{system of
    arcs}~\cite[Section~5.1]{ccelw-scsh-08}, namely, a set of disjoint,
  simple paths with endpoints on the boundary of~$S'$ that cut~$S'$ into a
  topological disk.  Moreover, there are $O(g+t)$ paths (by Euler's
  formula), and the aforementioned construction guarantees that each path
  is roughly the ``concatenation'' of (at most) two shortest paths.  More
  precisely, each path can be split into two shortest paths by inserting a
  degree-two vertex on an intersection point between $e$ and some edge
  of~$G'$.  (It is actually a \emph{shortest} system of
  arcs~\cite[Conclusion]{c-scgsp-10}.)

  Putting back the disks containing the terminals, and extending the arcs
  slightly inside these disks to the terminals, we obtain a cut graph~$K$
  satisfying the desired properties.  The number of edges of~$K$ is still
  $O(g+t)$, and the running time is
  $O(n\log n+(g+t)n)$~\cite{ccelw-scsh-08}.
\end{proof}

At a high level, the algorithm consists in (1) enumerating all possible
``topologies'' of the multicut dual with respect to~$K$, (2) for each of
these possible topologies, computing a shortest multicut dual with that
topology, and (3) returning the overall shortest multicut dual.

\section{Structural Properties of a Shortest Multicut Dual}\label{S:struct}

In this section, we prove some structural properties of a shortest multicut
dual.

Consider all shortest multicut duals in general position with respect to
$K\cup G$.  Among all these, let $C_0$ be one that crosses~$K$ a minimum
number of times.  We can, without loss of generality, assume that $C_0$ is
inclusionwise minimal, in the sense that no edge can be removed from~$C_0$
without violating the fact that it is a multicut dual.  Of course, we can
assume that $C_0$ has no isolated vertex.  If $C_0$ has a degree-one
vertex, we can ``prune'' it, by removing it together with its incident
edge.  If $C_0$ has a degree-two vertex that is not a loop, we can
``dissolve'' it, by removing it and identifying the two incident edges.
Thus, we can assume that $C_0$ has minimum degree at least two, and that
every degree-two vertex is the vertex of a connected component that is a
loop.

\subsection{Crossing Bound}

We start with an easy consequence of Euler's formula.
\begin{lemma}\label{L:sizedual}
  $C_0$ has $O(g+t)$ vertices and edges.
\end{lemma}
\begin{proof}
  We first note that each face of~$C_0$ contains at least one terminal:
  Otherwise, let $e$ be an edge incident to a face not containing a
  terminal; we could remove~$e$ without violating the multicut dual
  property (because it would not change which pairs of terminals are
  separated), contradicting the minimality of~$C_0$.  Thus, $C_0$ has at
  most $t$ faces.

  Let us first assume that every face of~$C_0$ is a disk.  If $C_0$
  contains a degree-two vertex, this means that some connected component
  of~$C_0$ is a loop~$\ell$.  Since each face of~$C_0$ is a disk and $S$ is
  connected, the graph~$C_0$ must be connected, so it equals~$\ell$, and
  the statement of the lemma holds.  So we can assume that $C_0$ has
  minimum degree at least three.  Let $v$, $e$, and~$f$ be the numbers of
  vertices, edges, and faces of~$C_0$; we deduce that $3v\le2e$.  Combining
  this with Euler's formula $v-e+f=2-g$ and the fact that $f\le t$, we
  obtain that $e=O(g+t)$, and thus also $v=O(g+t)$, as desired.

  If some faces of~$C_0$ are not disks, we can iteratively extend~$C_0$ by
  adding edges between existing vertices of~$C_0$ so that no new face is
  created but each face is cut into a disk (see, e.g., Chambers et
  al.~\cite[proof of Lemma~2.1]{ccelw-scsh-08}; the proof in that article
  extends verbatim to non-orientable surfaces).  Applying the reasoning of
  the previous paragraph to this new graph yields that it has $O(g+t)$
  vertices and edges; this is also true for~$C_0$.
\end{proof}

The main structural property of~$C_0$ is isolated in the following lemma:
\begin{lemma}\label{L:crossingbound}
  There are $O(g+t)$ crossings between~$C_0$ and each edge of~$K$.
\end{lemma}
As in algorithms for other problems using the same
approach~\cite{km-spktc-12,ccelw-scsh-08,en-sncwp-11,cen-mcshc-09,f-sntcd-13},
the proof of this lemma consists of an exchange argument: If $C_0$ crosses
an edge of~$K$ too many times, we can replace $C_0$ with a no longer
multicut dual that crosses~$K$ fewer times, contradicting the choice
of~$C_0$.  The proof ultimately boils down to topological considerations.
Let us also mention that the only property that we are using on the edges
of~$K$ is that they are disjoint shortest paths (except possibly at their
endpoints).
\begin{proof}[Proof of Lemma~\ref{L:crossingbound}]
  We focus on a specific edge~$e$ of~$K$ crossed by~$C_0$, forgetting about
  the others.  It is convenient to put an \emph{obstacle} close to each of
  the two endpoints of~$e$ (since $e$ is a shortest path, its endpoints are
  distinct).  It is also convenient to temporarily look at the situation
  differently, by forgetting about~$G$ and by modifying $C_0$ in the
  vicinity of~$e$ by pushing all crossings of~$C_0$ with~$e$ to a single
  point~$p$ on~$e$ (Figure~\ref{F:contract-exchange-monogon}(a, b)).  This
  transforms~$C_0$ into another graph~$C'_0$ that has~$p$ as a new vertex.
  To prove the lemma, it suffices to prove that the degree of~$p$ in~$C'_0$
  is $O(g+t)$.  Moreover, every non-loop edge of~$C'_0$ corresponds to one
  of the two endpoints of an edge of~$C_0$, and there are $O(g+t)$ of these
  by Lemma~\ref{L:sizedual}.  Hence, if we let $L$ be the one-vertex
  subgraph of~$C'_0$ made of the loops of~$C'_0$ based at~$p$, it suffices
  to prove that the number of loops in~$L$ is $O(g+t)$.

  \begin{figure}
    \centerline{%
  \begin{inkfragenv}\def\svgwidth{\linewidth}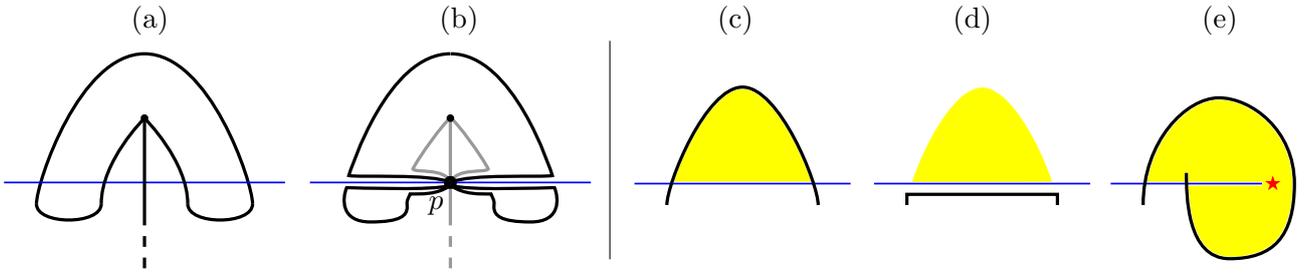\end{inkfragenv}%
}
    \caption{(a): The part of the multicut dual~$C_0$ close to~$e$
      (depicted as a horizontal line).  (b): its modified version~$C'_0$
      obtained by pushing all crossings with~$e$ to a single point~$p$.
      The black lines are the loops in~$L$, the grey ones are the other
      edges of~$C'_0$.  (c): The configuration corresponding to a monogon.
      The disk is shaded.  (d): The new configuration to replace~(c).  (e):
      This configuration is not a monogon, because of the presence of an
      obstacle (shown as a star).}
    \label{F:contract-exchange-monogon}
  \end{figure}

  A \emph{monogon}, resp.\ a \emph{bigon}, is a face of~$L$ that is,
  topologically, an open disk with one, resp.\ two, copies of~$p$ on its
  boundary and containing in its interior no obstacle, no vertex of~$C_0$,
  and no terminal.  We first claim that no face of~$L$ can be a monogon.
  Otherwise (Figure~\ref{F:contract-exchange-monogon}(c)), some edge~$e'$
  of~$C_0$ crosses~$e$ twice consecutively, at points $x$ and~$y$ say, such
  that the pieces of $e$ and~$e'$ between $x$ and~$y$ bound a disk
  containing in its interior no obstacle, no vertex of~$C_0$, and no
  terminal.  Since the disk contains no obstacle, the boundary of the disk
  lies entirely on one side of~$e$, as in
  Figure~\ref{F:contract-exchange-monogon}(c), and other cases such as the
  one shown in Figure~\ref{F:contract-exchange-monogon}(e) cannot occur.
  Since the disk contains no vertex of~$C_0$, it contains no piece of~$C_0$
  in its interior.  We can thus replace the piece of~$e'$ between $x$
  and~$y$ with a path that runs along~$e$
  (Figure~\ref{F:contract-exchange-monogon}(d)).  This operation does not
  make $e'$ longer, since $e$ is a shortest path; it removes the two
  intersection points with~$e$ and does not introduce other crossings
  with~$K$.  Moreover, since the disk contains no terminal in its interior,
  the resulting graph is also a multicut dual.  This contradiction with the
  choice of~$C_0$ proves the claim.
 
  We will prove below that no loop in~$L$ can be incident to two bigons.
  Taking this fact for granted for now, whenever one face of~$L$ is a
  bigon, we remove one of the two incident loops, and iterate until there
  is no bigon any more.  The previous fact implies that these iterations
  remove at most half of the loops: If $L'$ is the remaining set of loops,
  we have $|L|\le2|L'|$.  Furthermore, $L'$~has no monogon or bigon.  This
  latter fact, together with arguments based on Euler's formula, implies
  that the number of loops in~$L'$ is
  $O(g+t)$~\cite[Lemma~2.1]{ccelw-scsh-08}, because $S$ has Euler
  genus~$g$, and the total number of obstacles, vertices of~$C_0$, and
  terminals (which are the points that prevent a face that is a disk of
  degree one or two to be a monogon or bigon) is $O(g+t)$
  (Lemma~\ref{L:sizedual}).  (That article considers the orientable case
  only, but the lemma~\cite[Lemma~2.1]{ccelw-scsh-08} and its proof extend
  directly to the non-orientable case.)  This implies that $|L|=O(g+t)$,
  which proves the lemma.

  So there only remains to prove that no loop in~$L$ can be incident to two
  bigons.  Assume that such a loop exists.  On the original surface~$S$,
  this corresponds to two ``strips'' glued together, see
  Figure~\ref{F:exchange-orient}, top: Each strip is bounded by two pieces
  of~$e$ and two pieces of edges of~$C_0$, and these two strips share a
  common piece of edge of~$C_0$.  Since a bigon contains no obstacle, the
  sides of the strips contain none of the endpoints of~$e$.  Since the
  interiors of these strips contain no vertex of~$C_0$, they contain no
  piece of~$C_0$.

  \begin{figure}
    \centerline{%
  \begin{inkfragenv}\def\svgwidth{\linewidth}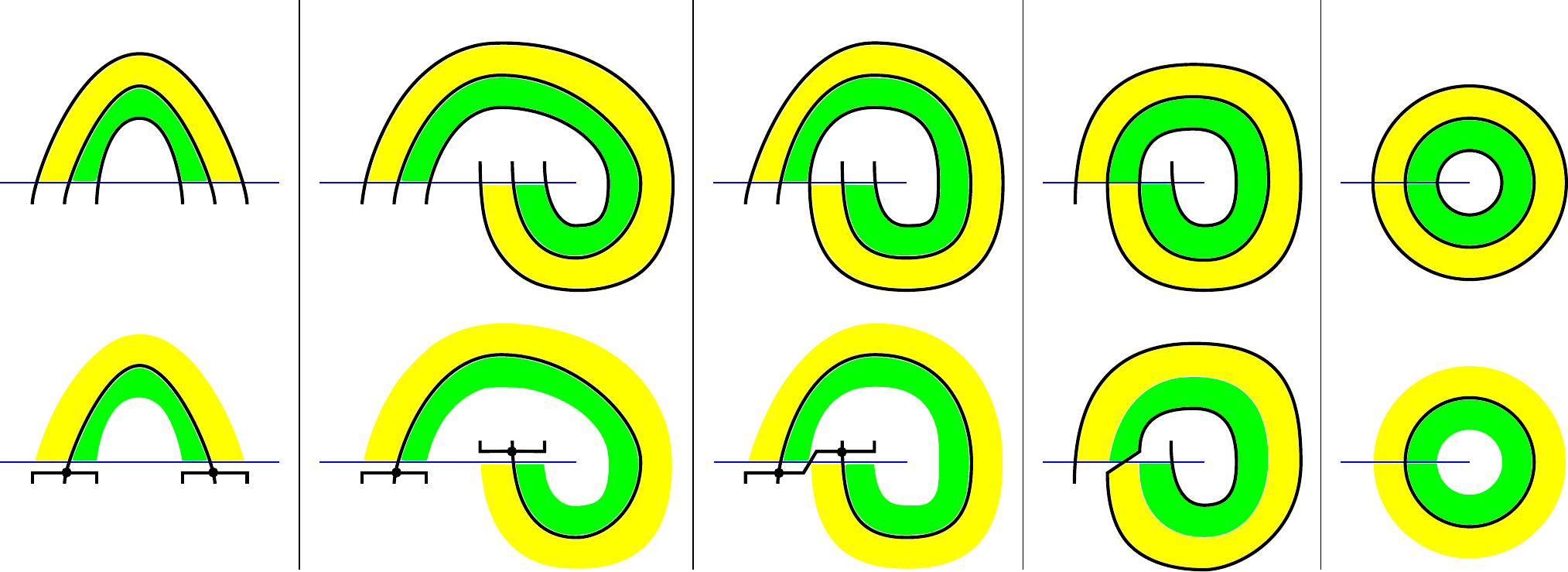\end{inkfragenv}%
}
    \caption[The exchange argument in the two-sided case.]{The exchange
      argument in the two-sided case.  The horizontal
      segment represents edge~$e$ of~$K$.  The strips are shaded; they
      represent disks with no terminal, no obstacle, and no piece of~$C_0$
      in their interior.

      If the sides of the strips are all on the same side of~$e$, there is
      a single case~(a).  Replacing the top configuration of~$C_0$ with the
      bottom configuration (creating two new vertices) still yields a
      multicut dual (as all pairs of faces that were separated in the top
      configuration are still separated in the bottom configuration, except
      possibly for the strips, but these contain no terminal), which is no
      longer than the original (because $e$ is a shortest path) and has
      less crossings with~$K$.  This is a
      contradiction with the choice of~$C_0$.

      If the sides of the strips are on different sides of~$e$, we need to
      distinguish according to four cases (b--e), depending on how the
      sides of the strips overlap.  In all cases, the same argument shows
      that we could find a no longer multicut dual with fewer crossings
      with~$K$, a contradiction.  (We could also remark that case~(e) is
      impossible because it involves closed curves in~$C_0$ without
      vertex.)}
    \label{F:exchange-orient}
  \end{figure}
  \begin{figure}
    \centerline{%
  \begin{inkfragenv}\def\svgwidth{\linewidth}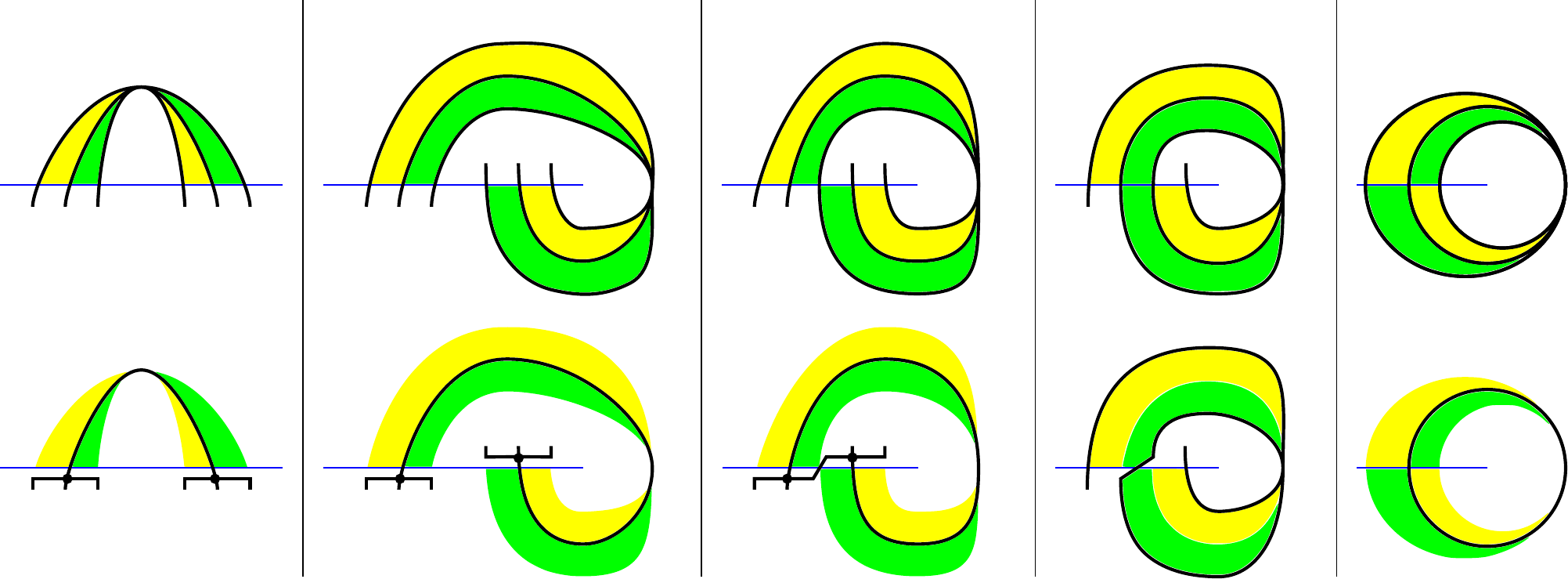\end{inkfragenv}%
}
    \caption{The exchange argument in the one-sided case.  The argument is
      the same as in Figure~\ref{F:exchange-orient}; the sole difference is
      that the strips are drawn with a ``twist''.}
    \label{F:exchange-nonorient}
  \end{figure}
 
  If $S$ is assumed to be orientable, there are five possible cases up to
  symmetry, see Figure~\ref{F:exchange-orient}, top: (a) is the case where
  each strip has its two sides on the same side of~$e$, (b--e) are the
  various cases where each strip has its two sides on opposite sides
  of~$e$.  In each case, we change $C_0$ by modifying some edges and
  possibly by adding vertices (see Figure~\ref{F:exchange-orient}, bottom).
  Since $e$ is a shortest path and the new pieces ``run along''~$e$, one
  can check that the resulting graph is no longer than~$C_0$; moreover, it
  crosses~$K$ fewer times.  Also, each replacement may split some faces of
  the original graph and attach each of the strips to some of the resulting
  faces, but pairs of terminals that were initially separated by~$C_0$ are
  still separated by the new graph, which is thus also a multicut dual.
  This contradicts the choice of~$C_0$.

  If $S$ is non-orientable, there are five other cases, because the loops
  in~$L$ may be one-sided.  However, an entirely similar argument as in the
  previous paragraph (Figure~\ref{F:exchange-nonorient}) allows to conclude.
\end{proof}

\subsection{Some Shortest Multicut Dual is Good}

We now give a more precise description of the intersection pattern
between~$C_0$ and~$K$, using the properties proved in the previous section.
Cutting the surface~$S$ along~$K$ yields a topological
disk~\emphdef{$\bm{D}$}.  The boundary~\emphdef{$\bm{\partial D}$} of~$D$
corresponds to copies of vertices and edges of~$K$; there are $O(g+t)$ of
these.  The copies of the vertices of~$K$ on~$\partial D$ are called the
\emphdef{corners} of~$D$, while the copies of the edges of~$K$ are called
the \emphdef{sides} of~$D$.  The sides of~$D$ can be reglued pairwise to
obtain~$S$.

\begin{figure}
  \centerline{%
  \begin{inkfragenv}\def\svgwidth{.95\linewidth}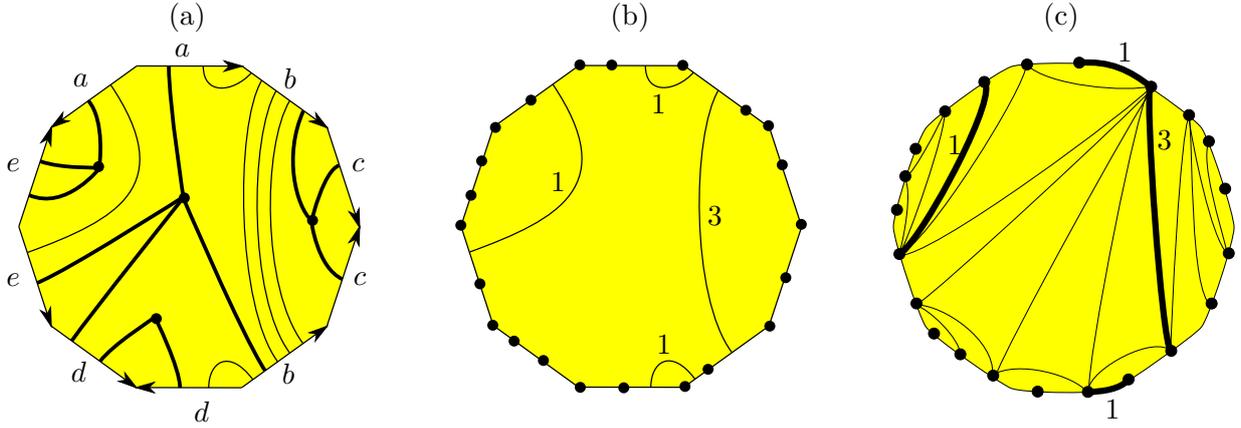\end{inkfragenv}%
}
  \caption{(a): The surface~$S$, and a (good) graph~$C$, viewed after
    cutting~$S$ along the cut graph~$K$, resulting in a disk~$D$.  The
    trees of~$\CC$ are in thick lines, the arcs being in thin lines.  (b)
    and (c) are used to illustrate the proof of
    Proposition~\ref{P:enumtop}: (b) shows the disk~$D$, with vertices on
    the boundary corresponding to the corners and the leaves of trees
    of~$\CC$; arcs connecting the same pair of edges of this polygon are
    represented by a single arc with an integer indicating their number;
    (c) represents the dual polygon (compared from~(b), every vertex on the
    boundary is replaced with an edge, and every edge is replaced with a
    vertex); the arcs correspond to the internal or boundary edges of this
    dual polygon in thick lines; these edges can be augmented arbitrarily
    to a weighted triangulation of the dual polygon (the weights of the
    edges are zero unless noted otherwise).}
  \label{F:good}
\end{figure}
Let $C$ be a graph on~$S$ in general position with respect to~$K\cup G$.
In particular, $C_0$ is such a graph.  Cutting $S$ along~$K$ transforms the
overlay of~$K$ and~$C$ into a graph~$U$ drawn in~$D$
(Figure~\ref{F:good}(a)): Each edge of that graph corresponds to a piece of
an edge of~$K$ or of~$C$; each vertex of that graph corresponds to a vertex
of~$K$, a vertex of~$C$, or a (four-valent) intersection point between an
edge of~$K$ and an edge of~$C$.  We denote by~\emphdef{$\bm\CC$} the
subgraph of~$U$ made of the edges corresponding to pieces of edges of~$C$
(thus, $\CC$ lies in the interior of~$D$ except possibly for some of its
leaves), and by~\emphdef{$\bm\KK$} the subgraph of~$U$ made of the edges
corresponding to pieces of edges of~$K$ (thus, the image of~$\KK$ is the
boundary of~$D$).

\begin{definition}\label{D:good}
  We say that $C$ is \emphdef{good} if $\CC$ is the disjoint union of (see
  Figure~\ref{F:good}(a)):
  \begin{itemize}
  \item \emphdef{trees} with at least one vertex of degree at least two,
    with all their leaves on~$\partial D$, and
  \item \emphdef{arcs}, namely, edges with both endpoints on~$\partial D$,
    on different sides of~$\partial D$,
  \end{itemize}
  and, moreover:
  \begin{itemize}
  \item there are $O(g+t)$ intersection points between $C$ and each side
    of~$\partial D$, and
  \item the total number of edges of the trees is $O(g+t)$.
  \end{itemize}
\end{definition}
Finally, Section~\ref{S:struct} can be summarized as follows:
\begin{proposition}\label{P:struct}
  Some shortest multicut dual is good.
\end{proposition}
\begin{proof}
  We prove that $C_0$ is indeed good.

  If some connected component of~$\CC_0$ contains a cycle, that cycle
  bounds a disk on~$S$ containing no terminal (since the terminals are
  vertices of~$K$), so we can remove any edge of that cycle without
  affecting the fact that $C_0$ is a multicut dual; this contradicts the
  choice of~$C_0$.  Therefore, $\CC_0$ is a forest.

  Since $C_0$ has no degree-zero or degree-one vertex, $\CC_0$ has no
  isolated vertex, and any degree-one vertex of~$\CC_0$ lies on~$\partial
  D$.  Any component of~$\CC_0$ that is a single path must have its
  endpoints on different sides of~$\KK$; otherwise, we could take an
  innermost such path and ``push it across'' the side of~$\CC_0$ its
  endpoints lie in, using the same argument as in
  Figure~\ref{F:contract-exchange-monogon}(a, b): This would not make $C_0$
  longer, and would decrease its number of crossings with~$K$,
  contradicting the choice of~$C_0$.  This proves that $C_0$~satisfies the
  first two points of Definition~\ref{D:good}.  The third and fourth points
  trivially follow from Lemmas \ref{L:crossingbound} and~\ref{L:sizedual},
  respectively.
\end{proof}

\section{Enumerating Topologies}\label{S:enum}

Let $C$ be a good graph on~$S$; recall that the union of~$\KK$ and~$\CC$
forms a connected planar graph~$U$.  The \emphdef{topology of~$\bm{C}$} is
the data of the combinatorial map of~$U$, where the outer face is
distinguished, and the sides are paired.  Intuitively, it describes
combinatorially the relative positions of~$C$ and~$K$.  More generally, a
\emphdef{topology} is the combinatorial map of a connected, planar graph
with a distinguished outer face and a pairing of the sides (these are the
subpaths of the outer cycle connecting consecutive degree-two vertices).

\begin{lemma}\label{L:algovalid}
  Given a topology, we can determine whether it is the topology of a good
  graph that is a multicut dual in $O((g+t)^2)$ time.
\end{lemma}
\begin{proof}
  First, we can assume that the topology has complexity $O((g+t)^2)$, since
  otherwise it is not the topology of a good graph.  We first check that
  the outer boundary is a cycle.  Then we glue the sides of the outer face
  according to the pairing.  This results in the combinatorial map, on~$S$,
  of the overlay of $C$ (corresponding to the interior edges of the
  topology) and~$K$ (corresponding to the edges on the outer face---we can
  check that the combinatorial map indeed is the same as that of~$K$).  All
  of this takes $O((g+t)^2)$ time.

  Deciding whether $C$ is good is easy in time linear in the size of the
  topology.  Deciding whether $C$ is a multicut dual can also be done in
  that amount of time: We can determine which elements of~$T$ fall into
  which face of~$C$, since each terminal is a vertex of~$K$, and whether
  each pair of terminals in~$R$ lies in different faces of~$C$.
\end{proof}
This also implies that whether $C$ is a multicut dual or not is completely
determined by its topology.  Hence the following terminology: A topology is
\emphdef{valid} if it is the topology of a good graph that is a multicut
dual.  The result of this section can now be stated as follows.
\begin{proposition}\label{P:enumtop}
  The number of valid topologies is ${(g+t)}^{O(g+t)}$; these topologies
  can be enumerated within the same time bound.
\end{proposition}
\begin{proof}[Proof of Proposition~\ref{P:enumtop}]
  It suffices to prove that there are ${(g+t)}^{O(g+t)}$ possible
  topologies for a good graph, and that we can enumerate them within the
  same time bound, because we can afterwards select the ones that
  correspond to multicut duals by Lemma~\ref{L:algovalid}.

  We first focus on enumerating all the possibilities for the trees
  of~$\CC$, forgetting about the arcs.  We know that the total number of
  edges of the forest is $O(g+t)$.  Any such forest can be obtained by:
  \begin{itemize}
  \item starting with a tree~$\tau$ with $O(g+t)$ edges where all internal
    vertices have degree three,
  \item contracting an arbitrary number of edges of that tree,
  \item removing an arbitrary number of edges of that tree.
  \end{itemize}
  Moreover, to specify exactly the topology for the trees of~$\CC$, we need
  to specify the cyclic ordering of the edges around each vertex of~$\tau$,
  and to specify to which side of~$\KK$ each leaf of~$\tau$ belongs.

  There are $2^{O(g+t)}$ possibilities for the choice of the initial
  tree~$\tau$ together with the cyclic ordering of the edges around each
  vertex of~$\tau$ (this is essentially a Catalan number, counting the
  number of binary trees).  Once this is determined, there are
  $(g+t)^{O(g+t)}$ possibilities to determine to which side of~$\KK$ each
  leaf of~$\tau$ belongs (actually, $2^{O(g+t)}$, but this refinement is
  useless here).  There remains to choose which of the $O(g+t)$ edges to
  contract or delete, and to specify to which of the $O(g+t)$ sides
  of~$\KK$ each leaf of~$\tau$ belongs.  To conclude, there are
  $(g+t)^{O(g+t)}$ possible choices for the combinatorial map for the union
  of the trees, which can be enumerated also in $(g+t)^{O(g+t)}$ time.

  Given such a possibility for the trees, we bound the number of
  possibilities for choosing the arcs.  The $O(g+t)$ leaves of the trees
  and the corners cut~$\partial D$ into $O(g+t)$ intervals.  The arcs
  connect distinct intervals (by definition of a good graph); moreover, the
  arcs are pairwise disjoint.  Consider the \emph{dual polygon} having one
  vertex per interval, such that each edge connects two consecutive
  intervals along~$\KK$ (Figure~\ref{F:good}(b, c)).  Connect two vertices
  of the dual polygon by an internal edge of the dual polygon whenever
  there exists at least one arc connecting the corresponding pair of
  intervals.  Since the arcs are non-crossing, the internal edges do not
  cross, and form a subset of a \emph{triangulation} of the dual polygon.
  There are $2^{O(g+t)}$ possible triangulations, since the dual polygon
  has $O(g+t)$ vertices (this is again a Catalan number).  To conclude, the
  number of possibilities for inserting the arcs, assuming the trees have
  already been chosen, is bounded from above by the number of
  \emph{weighted triangulations} of the dual polygon, namely,
  triangulations where each (internal or boundary) edge bears a
  non-negative number that encodes the number of arcs of the corresponding
  type, which is thus $O(g+t)$.  This number is $(g+t)^{O(g+t)}$.
  Moreover, enumerating all these possibilities can be done in the same
  amount of time.
\end{proof}

\section{Dealing With Each Valid Topology}\label{S:final}

The strategy for our algorithm is roughly as follows: For each valid
topology, we compute a shortest graph~$C$ with that topology; then we
return the set of edges of~$G$ crossed by the overall shortest graph~$C$.
Actually, we do not exactly compute an embedding of a shortest graph~$C$;
instead, we compute a shortest drawing (possibly with crossings) of the
same graph, with some homotopy constraints; in particular, that drawing is
no longer than the shortest embedding, and we prove that this also
corresponds to a minimum multicut.  To get an efficient algorithm, we use
dynamic programming on small treewidth graphs~\cite{b-dpgbt-88}.  The key
proposition is the following.
\begin{proposition}\label{P:enumvert}
  Given a valid topology, we can compute, in
  $n^{O\big(\sqrt{g^2+gt}\big)}$ time, some multicut whose
  weight is at most the length of each multicut dual with that topology.
\end{proposition}
The following lemma allows to use treewidth techniques for surface-embedded
graphs; it seems to be folklore, and is standard in the planar case.
\begin{lemma}\label{L:treewidth}
  Given a graph~$H$ with $p$ vertices, edges, and faces embedded on a
  surface of genus~$g$, one can compute in $O(p\sqrt p)$ time a path
  decomposition of the graph with width $O(\sqrt{gp})$.
\end{lemma}
\begin{proof}
  We start by removing loops and multiple edges, in time~$O(p)$.  Note that
  the resulting graph (still denoted by~$H$) is, as the original one, not
  necessarily cellularly embedded on the surface~$S$ of genus~$g$.
  However, the combinatorial map given by~$H$ defines a cellular embedding
  of~$H$ on the surface~$S'$ obtained by pasting a disk to each facial
  walk of~$H$ (this idea is recurrent, e.g., in Mohar and
  Thomassen~\cite{mt-gs-01}).  The genus~$g'$ of~$S'$ is at most that
  of~$S$.  (Indeed, $H$ could be augmented to a graph~$H'$ cellularly
  embedded on~$S$ by adding edges, without increasing the number of
  vertices and faces; the cellular embeddings $H'$ on~$S$ and $H$ on~$S'$
  have the same number of vertices and faces, but the first one has more
  edges; so $S$ has genus larger than~$S'$, by Euler's formula.)

  Let $h$ be the number of vertices of~$H$.  Given the embedding of~$H$
  on~$S'$, one can compute in $O(h)$ time a small, balanced, planar
  separator~$A$ for~$H$: a set of $O(\sqrt{gh})$ vertices whose removal
  leaves a planar graph with no connected component with more than $2h/3$
  vertices~\cite[Theorem~5.1]{e-dgteg-03}.  (That theorem states only the
  existence of the separator, but the proof immediately gives a linear-time
  algorithm, as mentioned in the remark right after its proof.)

  We then compute a path decomposition of the planar graph~$H':=H-A$ of
  width $O(\sqrt{h'})$ in $O(h'\sqrt{h'})$ time, where $h'$ is the number
  of vertices of~$H'$, using standard techniques.  For example, one can
  compute a small, balanced, planar separator~$B$ for $H'$ made of at most
  $c_1\sqrt{h'}$ vertices, for some constant~$c_1$~\cite{lt-stpg-79}; by
  induction, we can assume that each component of~$H'-B$ has a path
  decomposition of width at most $2c_2h'/3$, for some constant~$c_2$ to be
  chosen later; concatenating these paths arbitrarily and adding~$B$ to all
  the nodes of the resulting path gives a path decomposition of~$H'$ of
  width at most $2c_2h'/3+c_1\sqrt{h'}$, which is at most $2c_2h'$ if $c_2$
  is chosen large enough.  By induction, it takes $O\big(h'\sqrt{h'}\big)$
  time to compute this path decomposition explicitly.  (Using more advanced
  techniques~\cite{g-psppt-95}, and if one is willing to accept an implicit
  representation of the path decomposition, one can achieve the same result
  in linear time, but we do not care about this, since this will not
  improve the running time of the overall algorithm.)

  Finally, adding $A$ to each node of the path decomposition of $H'=H-A$
  gives a path decomposition of~$H$ of width $O(\sqrt{gp})$ in $O(p\sqrt
  p)$ time.
\end{proof}

Let $C$ be a good graph.  The \emphdef{crossing sequence} of an edge~$e$
of~$C$ (directed arbitrarily) is the ordered sequence of edges in~$K$
crossed by~$e$ when walking along~$e$, together with the indication of the
orientation of each crossing (more precisely, on which side of the edge
of~$K$ lies the part of~$e$ before the crossing, and on which side lies the
part of~$e$ after the crossing).  Given the topology of~$C$, one can
determine the crossing sequence of every edge of~$C$.  We say that a
drawing~$C'$ of the (abstract) graph~$C$ has the same \emphdef{topology}
as~$C$ if each edge of~$C'$ has the same crossing sequence as the
corresponding edge of~$C$.
\begin{lemma}\label{L:dualref}
  Let $C'$ be a drawing of a multicut dual~$C$ with the same topology
  as~$C$.  Then, the set of edges of~$G$ crossed by~$C'$ is a multicut.
\end{lemma}
\begin{proof}
  Let $\set{t_1,t_2}\in R$ be a pair of terminals, and let $f$ be the face
  of~$C$ containing~$t_1$.  The set of edges of~$C$ that are incident
  exactly once to the face of~$C$ containing~$t_1$ forms an \emph{even}
  subgraph~$C_1$ of~$C$, in which every vertex has even degree.  Moreover,
  $C_1$ separates $t_1$ from~$t_2$.  Let $C'_1$ be the drawing of the same
  subgraph in~$C'$.  To prove our result, it suffices to prove that $C'_1$
  also separates $t_1$ from~$t_2$, using the fact that the crossing
  sequences are the same in~$C_1$ and~$C'_1$.  Although intuitive, this
  fact is non-trivial, and its proof relies on two results involving
  homology.  We will use these results as a black box; in particular, no
  knowledge of homology is required here.

  We can assume that $C'$ has a finite number of self-intersection points,
  each of which is a crossing.  Therefore, $C'_1$ can be seen as an even
  graph embedded on~$S$ (by adding a new vertex at each crossing between
  two edges of~$C'_1$).  Our lemma is implied by two results by Chambers et
  al.~\cite{cen-mcshc-09}, reformulated here in our setting (that paper
  only considers orientable surfaces, but the two results we use extend
  immediately to non-orientable surfaces):
  \begin{itemize}
  \item if, for every edge~$e$ of~$K$, the even graphs $C_1$ and~$C'_1$
    cross $e$ with the same parity, then they are homologous
    (over~$\ZZ/2\ZZ$)~\cite[Lemma~3.4]{cen-mcshc-09}.  In our case, since the
    crossing sequences are equal, $C_1$ and~$C'_1$ are homologous;
  \item an even graph separates $t_1$ from~$t_2$ if and only if it is
    homologous, on the surface $S\setminus\{t_1,t_2\}$, to a small circle
    around~$t_1$~\cite[Lemma~3.1]{cen-mcshc-09}.  Thus, since $C_1$
    separates $t_1$ from~$t_2$, it is also the case for~$C'_1$.\qedhere
  \end{itemize}
\end{proof}

\begin{proof}[Proof of Proposition~\ref{P:enumvert}]
  Lemma~\ref{L:dualref} implies that it suffices to compute (the set of
  edges of~$G$ crossed by) a shortest drawing with the given topology.  

  Let us first explain how to achieve this, assuming that the locations of
  the vertices are prescribed.  Let $C$ be any multicut dual with that
  topology and these vertex locations; we do not know~$C$, but know the
  crossing sequence of its edges (they are determined by the topology) and
  the vertex locations (they are prescribed).  To achieve our goal, it
  suffices, for each edge~$e$ of~$C$, to compute a shortest path with the
  same crossing sequence as~$e$ and the same endpoints.  We remark that
  algorithms for computing a shortest homotopic path in $S$ minus the
  vertex set of~$K$ precisely achieve this goal~\cite{ce-tnpcs-10}: In
  short, we glue copies of the disk~$D$ according to the specified crossing
  sequence, in such a way that edge~$e$ ``lifts'' to that space~$\hat S_e$;
  then we compute a shortest path in~$\hat S_e$ connecting the same
  endpoints as that lift, and ``project'' back to~$S$.  (We can assign
  infinitesimal crossing weights to the edges of~$K$ to ensure that the
  crossing sequences of~$e$ and~$e'$ are the same.)  The complexity of~$D$
  (with its internal structure defined by the edges of~$G$) is $O((g+t)n)$,
  and the crossing sequences have total length $O((g+t)^2)$, so the total
  complexity of the spaces~$\hat S_e$ is $O((g+t)^3n)$.  Since $\hat S_e$
  is planar, and since shortest paths can be computed in linear time in
  planar graphs~\cite{hkrs-fspap-97}, this is also the complexity of
  computing (the set of edges of~$G$ crossed by) a shortest graph drawing
  with a given topology and specified vertex locations.

  To compute a shortest drawing with the given topology, over all choices
  of vertex locations, we can na\"\i{}vely enumerate all possible locations
  of the $O(g+t)$ vertices.  Note that it is only relevant to consider
  which face of the overlay of $G$ and~$K$ each vertex belongs to, and
  there are $O((g+t)n)$ such faces.  This yields an
  $n^{O(g+t)}$-time algorithm.  To get a better running time,
  we use treewidth techniques, also used by Klein and Marx in the planar
  multiway cut case~\cite{km-spktc-12}.  The (abstract) graph~$C$ defined
  by the specified topology has $O(g+t)$ vertices and is embedded on a
  surface with genus~$g$.  Lemma~\ref{L:treewidth} gives us, in
  $(g+t)^{O(1)}$ time, a path decomposition of~$C$ of width
  $O\big(\sqrt{g^2+gt}\big)$.  We use standard dynamic programming on the
  path decomposition (rooted, e.g., at one of its endpoints).  More
  precisely, at each node~$N$ of the path, we have a table that indicates,
  for every choice of the locations of the vertices in~$N$, the length of
  the shortest drawing of the subgraph of~$C$ induced by the vertices
  in~$N$ and its descendents, among those that respect the crossing
  sequence constraints.  We can fill in the tables by a bottom-up traversal
  of the path decomposition.  Since each node contains
  $O\big(\sqrt{g^2+gt}\big)$ vertices, the running time of the algorithm is
  $n^{O\big(\sqrt{g^2+gt}\big)}$.
\end{proof}
We can now conclude the proof of Theorem~\ref{T:main}:
\begin{proof}[Proof of Theorem~\ref{T:main}]
  We compute the cut graph~$K$ in $O(n\log n+(g+t)n)$ time
  (Proposition~\ref{P:cutgraph}), and enumerate all valid topologies in
  $(g+t)^{O(g+t)}$ time (Proposition~\ref{P:enumtop}).  For each valid
  topology, we apply the algorithm of Proposition~\ref{P:enumvert} in
  $n^{O\big(\sqrt{g^2+gt}\big)}$ time, and return a shortest multicut
  found.  Therefore, the overall running time is
  $O({(g+t)^{O(g+t)}}\*{n^{O\big(\sqrt{g^2+gt}\big)}}$.  The correctness is
  easy: By Proposition~\ref{P:dual}, it suffices to compute a multicut
  whose weight is at most the length of any multicut dual.  By
  Proposition~\ref{P:struct}, some shortest multicut dual has a valid
  topology; when this topology is chosen in the course of the algorithm,
  Proposition~\ref{P:enumvert} guarantees that we have computed a shortest
  multicut.
\end{proof}

\section*{Acknowledgments}

Many thanks to C\'edric Bentz for several helpful discussions on his
manuscript~\cite{b-ptapm-12}.  Thanks also to Claire Mathieu for an
inspiring discussion in the preliminary stage of this paper, to Arnaud de
Mesmay for a useful remark, and to the anonymous referees for their
detailed reports, one of them suggesting that an improvement might be
possible with the treewidth technique of Klein and Marx~\cite{km-spktc-12}.


\appendix

\section{Problems with Bentz' Approach in the Planar Case}\label{A:bentz}

Here, we show several flaws in Bentz' approach~\cite{b-ptapm-12} for the
planar version of the problem, which leaves little hope for repair.  Bentz
starts by guessing the \emph{clusters}, namely, the partition of the
terminals induced by the faces of the optimal multicut dual, and enumerates
all possibilities for the locations of the vertices of the multicut dual
that have degree at least three.  (He does so in the dual graph, while we
do it in the cross-metric setting to keep track of the crossings between
the multicut dual and the cut graph, but this is the same idea.)  Then, his
strategy is to consider all the possible homotopy classes of the closed
walks~$C_i$ that bound a face of the multicut dual containing the terminals
of a given cluster~$i$.  Unfortunately, it is dubious that this strategy
could work; at least, the enumeration of the homotopy classes has to be
substantially more complicated than described in his article, and one would
need to consider the homotopy classes of the edges of the multicut dual,
not only of the closed walks~$C_i$.  In more detail:
\begin{figure}[b]
  \centerline{%
  \begin{inkfragenv}\def\svgwidth{.8\linewidth}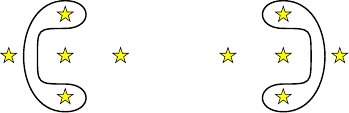\end{inkfragenv}%
}
  \caption{Two non-homotopic cycles enclosing the same set of
    terminals.}
  \label{F:counterex-lemma3}
\end{figure}
\begin{figure}
  \centerline{%
  \begin{inkfragenv}\def\svgwidth{.8\linewidth}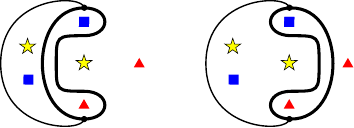\end{inkfragenv}%
}
  \caption{In a multicut dual, replacing a cycle (in thick lines) with
    another cycle passing through the same vertices and enclosing the same
    terminals does not always give a multicut dual.}
  \label{F:counterex-lemma5}
\end{figure}
\begin{itemize}
\item In Lemma~3, it is claimed that two cycles are homotopic in the plane
  minus a set of terminals if and only if they enclose the same terminals
  in their interior.  This is not the case (see
  Figure~\ref{F:counterex-lemma3}).  In particular, as soon as there is
  more than two terminals, the number of homotopy classes is infinite;
\item Lemma~5 claims the following: Assume that $C_i$ passes through
  vertices $\omega_1,\ldots,\omega_{q_i}$ of the multicut dual.  Then we
  still have a multicut dual if we replace $C_i$ with any cycle~$C'_i$
  going through $\omega_1,\ldots,\omega_{q_i}$ that encloses the same set
  of terminals as~$C_i$.  This is not the case, as
  Figure~\ref{F:counterex-lemma5} demonstrates.  (Actually, Lemma~5 makes a
  slightly stronger claim, which is also contradicted by that figure.)  For
  the same reason, Corollary~1 does not hold.
\end{itemize}
These counterexamples show that it is \emph{not} sufficient to determine
whether cycles are valid (can be part of a multicut dual) based only on
which terminals they enclose (although this would indeed be very nice,
since the number of such possibilities is $n^{O(t)}$).  We really need to
consider the homotopy classes of the paths, and to bound their number.
This is exactly the purpose of Lemma~\ref{L:crossingbound} in our
technique.

\end{document}